\documentclass[a4paper,USenglish,cleveref,autoref,thm-restate]{lipics-v2021}

\title{QPTAS for MWIS and finding large sparse induced subgraphs in graphs with few independent long holes}
\titlerunning{QPTAS for MWIS and stuff in graphs with few independent long holes}

\author{\'{E}douard Bonnet}{Univ Lyon, CNRS, ENS de Lyon, Université Claude Bernard Lyon 1, LIP UMR5668, France \and \url{http://perso.ens-lyon.fr/edouard.bonnet/}}{edouard.bonnet@ens-lyon.fr}{https://orcid.org/0000-0002-1653-5822}{}
\author{Jadwiga Czyżewska}{University of Warsaw, Poland}{j.czyzewska@mimuw.edu.pl}{https://orcid.org/0009-0006-0764-9613}{Supported by Polish National Science Centre SONATA BIS-12 grant number 2022/46/E/ST6/00143.}
\author{Tomáš Masařík}{University of Warsaw, Poland}{masarik@mimuw.edu.pl}{https://orcid.org/0000-0001-8524-4036}{Supported by the Polish National Science Centre SONATA-17 grant number
2021/43/D/ST6/03312.}
\author{Marcin Pilipczuk}{University of Warsaw, Poland}{}{https://orcid.org/0000-0001-5680-7397}{Supported by Polish National Science Centre SONATA BIS-12 grant number 2022/46/E/ST6/00143.}
\author{Paweł Rzążewski}{Warsaw University of Technology \& University of Warsaw, Poland}{}{https://orcid.org/0000-0001-7696-3848}{Supported by the National Science Centre grant 2024/54/E/ST6/00094.}

\authorrunning{\'E. Bonnet, J. Czyżewska, T. Masařík, M. Pilipczuk, P. Rzążewski}

\Copyright{\'Edouard Bonnet, Jadwiga Czyżewska, T. Masařík, M. Pilipczuk, P. Rzążewski}

\category{}

\relatedversion{}

\supplement{}

\acknowledgements{
The project was initiated at the Structural Graph Theory workshop STWOR in September 2023.
This workshop was a part of STRUG: Structural Graph Theory Bootcamp, funded by the ``Excellence initiative -- research university (2020--2026)'' of University of Warsaw.
We are grateful to Benjamin Bergougnoux, Linda Cook, and Marek Sokołowski for stimulating discussions on the topic.
}

\nolinenumbers 

\hideLIPIcs  

\usepackage[utf8]{inputenc}

\usepackage[T1]{fontenc}
\usepackage{lmodern}

\usepackage[colorinlistoftodos,bordercolor=orange,backgroundcolor=orange!20,linecolor=orange,textsize=normalsize]{todonotes}

\usepackage{amsmath}
\usepackage{mathtools}
\usepackage{amssymb}
\usepackage{amsthm}
\usepackage{bbm}
\usepackage{accents}
\usepackage{booktabs}
\usepackage{fixmath}
\usepackage[noadjust,compress]{cite}

\newtheorem{conj}{Conjecture}

\newtheorem{problem}{Problem}

\usepackage{xspace}
\usepackage{tikz}

\usetikzlibrary{fit}
\usetikzlibrary{arrows}
\usetikzlibrary{patterns}
\usetikzlibrary{calc}
\usetikzlibrary{shapes}
\usetikzlibrary{positioning}
\usetikzlibrary{math}
\usetikzlibrary{shapes.symbols}
\usetikzlibrary{decorations.pathreplacing,calligraphy}
\usepackage[scr=boondox,scrscaled=1.05]{mathalfa}

\newcommand{\MIS}{\textsc{MIS}\xspace}
\newcommand{\MWIS}{\textsc{MWIS}\xspace}

\renewcommand{\geq}{\geqslant}
\renewcommand{\leq}{\leqslant}
\newcommand{\wei}{\mathsf{w}}
\newcommand{\Oh}{\mathcal{O}}
\newcommand{\cL}{\mathcal{L}}
\newcommand{\cH}{\mathcal{H}}

\newcommand{\cX}{\mathcal{X}}
\newcommand{\cY}{\mathcal{Y}}
\newcommand{\cS}{\mathcal{S}}

\newcommand{\RR}{\mathbb{R}}
\newcommand{\NN}{\mathbb{N}}
\newcommand{\NP}{\textsf{NP}\xspace}

\usepackage{thmtools}
\usepackage{thm-restate}

\newcommand\tw{\textsf{tw}}

\usepackage{todonotes}

\usepackage{framed}
\usepackage{tabularx}

\newlength{\RoundedBoxWidth}
\newsavebox{\GrayRoundedBox}
\newenvironment{GrayBox}[1]%
   {\setlength{\RoundedBoxWidth}{.93\columnwidth}
    \def\boxheading{#1}
    \begin{lrbox}{\GrayRoundedBox}
       \begin{minipage}{\RoundedBoxWidth}}%
   {   \end{minipage}
    \end{lrbox}
    \begin{center}
    \begin{tikzpicture}%
       \node(Text)[draw=black!20,fill=white,rounded corners,inner sep=2ex,text width=\RoundedBoxWidth]
             {\usebox{\GrayRoundedBox}};
        \coordinate(x) at (current bounding box.north west);
        \node [draw=white,rectangle,inner sep=3pt,anchor=north west,fill=white]
        at ($(x)+(6pt,.75em)$) {\boxheading};
    \end{tikzpicture}
    \end{center}}

\newenvironment{defproblemx}[1]{\noindent\ignorespaces%
                                \FrameSep=6pt%
                                \parindent=0pt%
                \begin{GrayBox}{#1}%
                \begin{tabular*}{\columnwidth}{!{\extracolsep{\fill}}@{\hspace{.1em}} >{\itshape} p{1.5cm} p{0.86\columnwidth} @{}}%
            }{
                \end{tabular*}%
                \end{GrayBox}%
                \ignorespacesafterend
            }

\newcommand{\problemTask}[3]{%
  \begin{defproblemx}{#1}
    Input: & #2 \\
    Task: & #3
  \end{defproblemx}
}

\begin{document}

\maketitle

\begin{abstract}
We present a quasipolynomial-time approximation scheme (QPTAS) for the \textsc{Maximum Independent Set} (\textsc{MWIS}) in graphs with a bounded number of pairwise vertex-disjoint and non-adjacent long induced cycles.
More formally, for every fixed $s$ and $t$, we show a QPTAS for \textsc{MWIS} in graphs that exclude $sC_t$ as an induced minor.
Combining this with known results, we obtain a QPTAS for the problem of finding a largest induced subgraph of bounded treewidth with given hereditary property definable in Counting Monadic Second Order Logic, in the same classes of graphs.

This is a step towards a conjecture of Gartland and Lokshtanov which asserts that for any planar graph $H$,
graphs that exclude $H$ as an induced minor admit a polynomial-time algorithm for the latter problem.
This conjecture is notoriously open and even its weaker variants are confirmed only for very restricted graphs $H$.
\end{abstract}

\section{Introduction}\label{sec:intro}

In the \textsc{Max Independent Set} (\MIS) problem, one is asked, given a graph $G$, for a~largest \emph{independent set}, i.e., a~set of pairwise nonadjacent vertices in $G$.
In the weighted variant of the problem, \textsc{Max Weight Independent Set} (\MWIS), the input graph has vertex weights and we ask for an independent set of maximum weight.
\MIS (and thus \MWIS) is a ``canonical'' hard problem---it is one of Karp's 21 \NP-hard problems~\cite{Karp1972}, \textsf{W[1]}-hard (with respect to the solution size)~\cite{platypus} and notoriously hard to approximate~\cite{Hastad96cliqueis}, even in the parameterized setting~\cite{DBLP:journals/siamcomp/ChalermsookCKLM20,DBLP:conf/focs/LinRSW23}.
Thus, a natural question to ask when dealing with such a hard problem is as follows: 
\emph{For what classes of input graphs is the \MWIS problem tractable?}

Typically, this question is studied for \emph{hereditary} graph classes, i.e., classes closed under vertex deletion. Each such class can be equivalently characterized by specifying minimal induced subgraphs that do not belong to the class. For graphs $G$ and $H$, we say that $G$ is \emph{$H$-free} if it does not contain $H$ as an induced subgraph. 

The complexity study of \MIS and \MWIS in restricted graph classes is among the most active research directions in algorithmic graph theory. Let us list some relevant results, focusing on the case of $H$-free graphs, for a single graph $H$.
Already in the early 1980s, Alekseev~\cite{alekseev1982effect} observed that classic \NP-hardness reductions imply that \MIS remains \NP-hard in $H$-free graphs for many graphs $H$.
First, \MIS is \NP-hard in subcubic graphs, and thus in $H$-free graphs whenever $H$ has a vertex of degree at least 4~\cite{GAREY1976237}.
The second reduction involves the so-called \emph{Poljak's subdivision trick}~\cite{Po74}: 
subdividing an edge of a graph twice yields a new graph where the size of a maximum independent set increases by exactly one.
Consequently, for a fixed graph $H$, we can start with an arbitrary instance of  \MIS and subdivide each edge $2|V(H)|$ times.
This way we obtain an equivalent instance that is $H$-free, whenever $H$ has a cycle or two vertices of degree at least three in the same connected~component.

Combining these two observations, we obtain that \MIS (and thus, \MWIS) remains \NP-hard in $H$-free graphs, unless every component of $H$ is a subcubic tree with at most one vertex of degree 3; let us call the family of such forests $\mathcal{S}$. Let us highlight that in particular $\mathcal{S}$ contains all paths and, more generally, linear forests (i.e., forests of paths).

We do not know any \NP-hardness result for \MWIS in $H$-free graphs when $H \in \mathcal{S}$.
On the other hand, polynomial-time algorithms are known only for small graphs $H \in \mathcal{S}$~\cite{ALEKSEEV20043,DBLP:journals/jda/LozinM08,MINTY1980284,SBIHI198053,DBLP:conf/soda/LokshantovVV14,grzesik2020polynomialtime,DBLP:journals/dam/BrandstadtM18a}.
However, general belief in the community is that all cases not excluded by standard \NP-hardness reductions mentioned above \MWIS is polynomial-time-solvable.
This belief is supported by the existence of a \emph{quasipolynomial-time} algorithm.
Indeed, for every $H \in \mathcal{S}$, the \MWIS problem restricted to $n$-vertex $H$-free graphs can be solved in time $n^{\Oh(\log^{16}n)}$~\cite{DBLP:conf/stoc/GartlandLMPPR24}.
Note that this is a strong indication that no $H \in \mathcal{S}$ defines an \NP-hard case of \MWIS, as otherwise every problem in \NP can be solved in quasipolynomial time.

Let us have a closer look at the case of $P_t$-free graphs, i.e., graphs that exclude a $t$-vertex path as an induced subgraph.
It turns out that in these classes we cannot only solve \MWIS in quasipolynomial time~\cite{DBLP:conf/focs/GartlandL20,DBLP:conf/sosa/PilipczukPR21} (and in polynomial time for $t \leq 6$~\cite{DBLP:conf/soda/LokshantovVV14,grzesik2020polynomialtime}), but the algorithms can actually be extended to a rich family of problems defined as follows~\cite{Tara,DBLP:conf/soda/ChudnovskyMPPR24,DBLP:conf/stoc/GartlandLPPR21}.
For an integer $r$ and a \textsf{CMSO}$_2$ sentence\footnote{\textsf{CMSO}$_2$ is a logic where one can use vertex, edge, and (vertex or edge) set variables, check vertex-edge incidence, quantify over variables, and apply counting predicates modulo fixed integers. See~\cite{platypus} for a formal introduction.} $\psi$, we define the $(\tw \leq r,\psi)$-\MWIS problem as follows (here, \MWIS stands for \emph{maximum-weight induced subgraph}).

\problemTask{$(\tw \leq r,\psi)$-\MWIS}%
{A graph $G$ equipped with a weight function $\wei\colon V(G) \to \RR_{\geq 0}$.}%
{Find a set $S \subseteq V(G)$, such that
\begin{itemize}
\item $G[S] \models \psi$,
\item $\tw(G[S]) \leq r$, and
\item $S$ is of maximum weight subject to the conditions above,
\end{itemize}
or conclude that no such set exists.}

Notable special cases of the $(\tw \leq r,\psi)$-\MWIS problem are \MWIS, \textsc{Feedback Vertex Set} (equivalently, \textsc{Max Induced Forest}), and \textsc{Even Cycle Transversal}  (equivalently, \textsc{Max Induced Odd Cactus}).  

Interestingly, the quasipolynomial-time algorithm for $(\tw \leq r,\psi)$-\MWIS in $P_t$-free graphs can be extended to the class of $C_{\geq t}$-free graphs: ones that do not contain an induced cycle with at least $t$ vertices.
Note that every $P_{t}$-free graph is $C_{\geq t+1}$-free.
We emphasize that $C_{\geq t}$-free graphs form the class defined by an infinite minimal family of forbidden induced subgraphs.
Furthermore, we cannot hope for tractability of $(\tw \leq r,\psi)$-\MWIS in $H$-free graphs, when $H$ is a single graph other than a linear forest.
Indeed, recall that if $H \notin \cS$, then already \MWIS ($r=0$) is \NP-hard in $H$-free graphs.
On the other hand, if $H \in \cS$ but is not a linear forest, then it must contain the \emph{claw}---the three-leaf star,
and \textsc{Max Induced Forest} ($r=1$) is \NP-hard in the class of claw-free graphs~\cite{DBLP:journals/dm/Munaro17}.

This bring us to a question:
What is the crucial property of $P_t$-free graphs that also extends to $C_{\geq t}$-free graphs, but not to $H$-free graphs for any fixed $H$ that is not a linear forest, and allows us to solve $(\tw \leq r,\psi)$-\MWIS efficiently?

It occurs that we should be looking at these classes from a different angle.
For graph $G$ and $H$, we say that $H$ is an \emph{induced minor} of $G$ if it can be obtained from $G$ by deleting vertices and contracting edges.
Equivalently, this means that there is an \emph{induced minor model} of $H$ in $G$: a collection of $|V(H)|$ pairwise disjoint subsets of $V(G)$, each inducing a~connected graph, and a bijection that maps these sets the vertices of $H$ so that there is an edge between two sets if and only if their corresponding vertices are adjacent in $H$.
We say that $G$ is \emph{$H$-induced-minor-free} if it does not contain $H$ as an induced minor.
Note that every $H$-induced-minor-free graph is in particular $H$-free. However, if $H$ is a linear forest, then $H$-free graphs and $H$-induced-minor-free graphs coincide. Furthermore, $C_{\geq t}$-free graphs are precisely $C_t$-induced-minor-free graphs.

The following conjecture asserts that the tractability of  $(\tw \leq r,\psi)$-\MWIS actually extends to all classes excluding a fixed planar graph as an induced minor.

\begin{conj}[Gartland, Lokshtanov~\cite{GartlandThesis}]\label{conj:alg}
For every planar graph $H$, every problem expressible as $(\tw \leq r,\psi)$-\MWIS is polynomial-time-solvable in $H$-induced-minor-free graphs.
\end{conj}

Let us emphasize that the assumption that we exclude a planar graph is crucial in \cref{conj:alg}.
Indeed, the class of planar graphs excludes any non-planar graph as an induced minor and \MWIS is \NP-hard in planar graphs~\cite{GAREY1976237}.

We are still very far from confirming \cref{conj:alg} in full generality.
Indeed, even its weakening asking for a quasipolynomial-time algorithm, or even a quasipolynomial-time approximation scheme (QPTAS) seems challenging.
Still, in recent years some special cases of \cref{conj:alg} were shown~\cite{DBLP:conf/soda/BonamyBDEGHTW23,DBLP:journals/algorithmica/BonnetDGTW26,dallard2022secondpaper}.
In the context of the current paper, the most relevant one seems to be the result of Bonnet et al.~\cite{DBLP:journals/algorithmica/BonnetDGTW26} that \MWIS can be solved in quasipolynomial time in the class of $(C_4 + sC_3)$-induced-minor-free graphs, for every fixed $s$.
Here, by ``$+$'' we mean disjoint union and multiplication by $s$ means $s$-fold disjoint union.
Thus, $C_4 + sC_3$ is the graph with $s+1$ components: one being a $C_4$ and the remaining ones being triangles.

Thus, in quasipolynomial time we can solve \MWIS (and even $(\tw \leq r,\psi)$-\MWIS) in graphs that exclude \emph{long} induced cycles (i.e., $C_{\geq t}$-free) and in graphs that exclude \emph{many} induced and pairwise non-adjacent cycles (i.e., $(C_4 + sC_3)$-induced-minor-free graphs).
In this paper, we are interested in the common generalization of these classes: graphs that exclude \emph{many} induced pairwise non-adjacent \emph{long} cycles, i.e., the class of $sC_t$-induced-minor-free graphs, for fixed $s$ and $t$.

While we are not able to prove \cref{conj:alg} for $sC_t$-induced-minor-free graphs, we provide a major step towards such a result.
As the main technical contribution, we show a~QPTAS for \MWIS in the considered classes of graphs.

\begin{restatable}{theorem}{thmtCt}\label{thm:mwis}
Let $s,t$ be positive integers and $\varepsilon \in (0,1)$ be a real.
There is an algorithm that, given a vertex-weighted graph $G$, in quasipolynomial time returns either:
\begin{itemize}
\item an induced minor model of $sC_t$ in $G$, or
\item an independent set of weight at least $(1-\varepsilon)$ times the maximum possible weight of an independent set in $G$.
\end{itemize}
\end{restatable}

\medskip

Next, combining \cref{thm:mwis} with known results concerning approximation schemes~\cite{DBLP:conf/stoc/GartlandLPPR21}, we obtain a QPTAS for the unweighted variant of $(\tw \leq r,\psi)$-\MWIS under an additional mild assumption that $\psi$ is a \emph{hereditary} \textsf{CMSO}$_2$ formula:
($i$) if $G \models \psi$, then $G' \models \psi$ for every induced subgraph $G'$ of $G$, and
($ii$) if $G_1 \models \psi$ and $G_2 \models \psi$, then $G_1 + G_2 \models \psi$.
We note that many natural graph properties, like e.g., planarity, bounded degeneracy, or excluding a fixed graph as a minor, can be defined by hereditary \textsf{CMSO}$_2$ formulas.

\begin{restatable}{theorem}{thmpack}\label{thm:pack}
Let $r \geq 0$, let $s,t$ be positive integers, $\varepsilon \in (0,1)$ be a real, and $\psi$ be a hereditary \textsf{CMSO}$_2$ formula.
There is an algorithm that, given a graph $G$, in quasipolynomial time returns one of the following outputs:
\begin{itemize}
\item an induced minor model of $sC_t$ in $G$, or
\item a solution to $(\tw \leq r,\psi)$-\MWIS of size at least $(1-\varepsilon)$ times the optimum one, or,
\item a correct conclusion that no solution to $(\tw \leq r,\psi)$-\MWIS exists.
\end{itemize}
\end{restatable}

\section{Preliminaries}

For integers $k, \ell$, the set $\{k, k+1, \ldots, \ell\}$ is denoted as $[k,\ell]$ and we shorten $[1,k]$ to $[k]$. 

We consider here (vertex-)weighted graphs $(G,\wei)$, where $\wei$ is a function $V(G)\to \RR_{\geq 0}$.
For a subset $X$ of vertices, we define its weight as $\wei(X) = \sum_{v \in X}\wei(v)$.
For a weighted graph $(G,\wei)$, by $\alpha(G,\wei)$ we denote the weight of a maximum-weight independent set in $G$. We assume that all computations on weights are performed in constant time.

The \emph{open neighborhood} of a vertex $v$ in $G$, denoted by $N_G(v)$, is the set of vertices adjacent to $v$.
The \emph{closed neighborhood} of $v$ is the set $N_G[v] = N_G(v) \cup \{v\}$.
For a subset $X$ of vertices of $G$, its \emph{open neighborhood} (resp., \emph{closed neighborhood}) we mean the set $N_G(X) = \bigcup_{v \in X} N_G(v) \setminus X$ (resp., $N_G[X] = \bigcup_{v \in X} N_G[v]$.
If $G$ is clear from the context, we omit the subscript and simply write $N(\cdot)$ and $N[\cdot]$.

We say that two disjoint sets $X,Y \subseteq V(G)$ are non-adjacent if there is no edge with one endpoint in $X$ and the other in $Y$.
Since all subgraphs in the paper are induced, we often identify such a subgraph with its vertex set.

For a connected graph $G$ and a set $A \subseteq V(G)$, the \emph{BFS-layering of $G$ from $A$} is the partition of $V(G)$ into sets $L_0,L_1,\ldots,L_h$ called \emph{layers}, where $L_0=A$ and for all $i \geq 1$, we have $L_i = N(L_{i-1}) \setminus \bigcup_{j < i-1} L_j$, and $h$ is the largest possible so that $L_h \neq \emptyset$.

A \emph{hole} in the graph is an induced cycle with at least 4 vertices.
The following easy observation allows us to look for holes of certain length.
\begin{lemma}\label{lem:find-hole}
    Given a graph $G$ and an integer $t \geq 4$,
    one can in time $\Oh(|V(G)|^t \cdot (|V(G)|+|E(G)|))$ find a
    shortest hole in $G$ of length at least $t$, or conclude that
    no such hole exists. 
\end{lemma}
\begin{proof}
First, in time $\Oh(|V(G)|^t)$ we exhaustively check whether $G$ has an induced cycle with exactly $t$ vertices. If such a cycle exists, we return it as it is clearly shortest possible.

Next, for every tuple $\tau$ of $t$ distinct vertices of $G$
we check whether $\tau$ induces a path.
Suppose this is the case and let $x,y$ the endvertices of this path. We remove from $G$ the closed neighborhood of all internal vertices of the path, except for $x$ and $y$.
Finally, we search for a shortest path connecting $x$ and $y$ in the obtained graph; together with the vertices from $\tau$ it forms an induced cycle with at least $t$ vertices.
We return the shortest of all cycles found in this process, or report that no cycle was found.
\end{proof}

Finally, let us recall a result that is particularly useful for constructing QPTASes for \MWIS. A~vertex $v$ of a~weighted graph $(G,\wei)$ is \emph{$\gamma$-heavy}
with respect to a~set $I \subseteq V(G)$ if $\wei(N_G[v] \cap I) \geqslant \gamma \cdot \wei(I)$.

\begin{lemma}[{\cite[Lemma 4.1]{DBLP:journals/siamcomp/ChudnovskyPPT24}}]\label{lem:delete-heavy}
Let $(G,\wei)$ be a weighted graph on $n$ vertices and $\gamma \in (0,1)$ be a real number.
In time $n^{\Oh(\log n / \gamma)}$ we can enumerate a family $\mathcal{I}'$ of $n^{\Oh(\log n / \gamma)}$ independent sets in $G$, each of size at most $\lceil{\gamma^{-1} \log n}\rceil$, such that for every independent set $I$ there exists $I' \in \mathcal{I}'$
such that $I' \subseteq I$ and every $\gamma$-heavy vertex with respect to $I$ belongs to $N[I']$.
\end{lemma}

Intuitively,~\cref{lem:delete-heavy} is a quasipolynomial approximation-preserving reduction to instances without $\gamma$-heavy vertices:
we can exhaustively guess $I' \in \mathcal{I}'$ and delete $N[I']$ from the graph.

\section{QPTAS for \MWIS: Proof of \cref{thm:mwis}}\label{sec:mwis}
In this section, we present a~QPTAS for \MWIS in graphs excluding $sC_t$ as an induced minor.

\thmtCt*

\begin{proof}
Fix $s$ and $t$. Since the value of $t$ will be fixed throughout the whole proof,
by a \emph{long hole} we mean an induced cycle with at least $t$ vertices.
We denote such a long hole shortly by $C_{\geq t}$. 

Let $(G,\wei)$ be a weighted graph on $n$ vertices.
Let $\varepsilon > 0$ be the desired precision, i.e., we aim for an $(1-\varepsilon)$-approximation or for finding an induced subgraph isomorphic to $sC_{\geq t}$.

\subparagraph*{Strategy.}
The algorithm is a typical recursive branching algorithm.
Each recursive call is invoked for an induced subgraph $G'$ of $G$ and an integer $s' \leq s$;
the goal is to either exhibit an~induced $s' C_{\geq t}$ subgraph of $G'$ or 
an independent set $I'$ of weight close to $\alpha(G',\wei)$ (the actual error analysis is made formally later;
it is not just merely an $(1-\varepsilon)$-approximation to $\alpha(G',\wei)$).
The initial call is to $G' = G$ and $s' = s$. 

We set
\[\beta \coloneqq \frac{\varepsilon}{s + t+\log_{6/5} n}~\text{ and }~\gamma \coloneqq \frac{\beta^3}{1000t}.\]

The computation of one recursive call is embedded in the following claim.
\begin{claim}\label{clm:set_X}
	Given an induced subgraph $G'$ of $G$ and an integer $1 \leq s' \leq s$, 
	in time $n^{\Oh(\log n / \gamma)}$ one can either
	report an induced $s'C_{\geq t}$ subgraph in $G'$ or enumerate
	a family $\mathcal{X}$ of pairs $(X,J)$ such that:
	\begin{enumerate}
		\item for every $(X,J) \in \mathcal{X}$, we have
		$X \subseteq V(G')$ and $J$ is an independent set in $G'$
    with $N[J] \subseteq X$;\label{p:1}
		\item for every $(X,J) \in \mathcal{X}$, every connected
		component $D$ of $G'-X$ satisfies one of the following conditions:
			\begin{enumerate}
				\item there is a $C_{\geq t}$ in $G'$ which is disjoint and nonadjacent to $D$;\label{p:hole} or
				\item $D$ has at~most $\frac{5}{6} |V(G')|$~vertices.\label{p:size}
			\end{enumerate}
		\item for every independent set $I'$ in $G'$ of weight $\alpha(G',\wei)$,
		there exists $(X,J) \in \mathcal{X}$ with $J \subseteq I'$ and
		$\wei((X \setminus J) \cap I') \leq \beta \alpha(G',\wei)$;
		\item the size of $\mathcal{X}$ is bounded by 
		$n^{\Oh(\log n / \gamma)}$.
	\end{enumerate}
\end{claim}

We now argue how \cref{clm:set_X} yields Theorem~\ref{thm:mwis}.
Consider a recursive call with parameters $(G',s')$. 
If $s' = 0$, then we return $\varnothing$ as an induced $s'C_{\geq t}$. 
If $V(G') = \varnothing$, we return the empty (independent) set. 

Otherwise, we apply Claim~\ref{clm:set_X} to $G'$ and $s'$.
If an induced $s'C_{\geq t}$ is returned, we return it and conclude. 
Otherwise, we iterate over the obtained family $\mathcal{X}$.
For every $(X,J) \in \mathcal{X}$, 
let $\mathcal{D}_X$ be the set of connected components of $G'-X$.
For every $D \in \mathcal{D}_X$,
we recurse on $G_D' \coloneqq G'[D]$ and either $s_D' \coloneqq s'$ if $|D| \leq \frac{5}{6} |V(G')|$ (Property~\ref{p:size})
or $s_D' \coloneqq s'-1$ otherwise. 
If an induced $s_D'C_{\geq t}$ is returned in $G_D'$, we augment it
with a $C_{\geq t}$ non-adjacent to $D$ in $G'$ if $s_D' = s'-1$
(such a long hole can be found using~\cref{lem:find-hole}) and return.
Otherwise, if an independent set $I_D'$ is returned, we 
compute an independent set $I_{(X,J)} = \bigcup_{D \in \mathcal{D}_X} I_D' \cup J$
(Note that this is an independent set due to \hyperref[p:1]{condition $N[J] \subseteq X$}.)
Finally, if no $s'C_{\geq t}$ was returned for any $(X,J) \in \mathcal{X}$,
we return $I_{(X,J)}$ of maximum weight among all options $(X,J) \in \mathcal{X}$.
This completes the description of the algorithm.

\subparagraph*{Complexity analysis.}
Let us move to the analysis.
For a recursive call $(G',s')$, let
\[
	\mu(G', s') \coloneqq s' + \left\lfloor \log_{6/5} |V(G')| \right\rfloor.
\]
Observe that every recursive subcall $(G_D', s_D')$ satisfies $\mu(G_D',s_D') \leqslant \mu(G',s')-1$.
Hence, the depth of the recursion is bounded by $s + \lfloor \log_{6/5} n \rfloor$.
Since every recursive call results in $n^{\Oh(\log n / \gamma)}$ subcalls, the overall running time is as follows: \[
 n^{\Oh \left( \log^2 n /\gamma \right)} =  n^{\Oh \left( \log^2 n/\beta^3 \right)} = n^{\Oh \left( \log^5 n/\varepsilon^3 \right)},
\]
i.e., quasipolynomial in $n$ (we hide terms depending on $s$ and $t$ in the $\Oh(\cdot)$-notation).

\subparagraph*{Correctness and approximation guarantee.}
Clearly, if any recursive call $(G',s')$ finds an induced $s'C_{\geq t}$,
it is propagated up in the recursion tree and results in exhibiting
an induced $sC_{\geq t}$ in the root call. Assume then that every
recursive call $(G',s')$ returned an independent set, which we denote
$I_{G',s'}$. 

Consider a recursive call $(G',s')$ and let $I'$ be an independent set in $G'$
of weight $\alpha(G',\wei)$. By the promise of Claim~\ref{clm:set_X},
there exists $(X,J) \in \mathcal{X}$ in this recursive call with
$J \subseteq I'$ and $\wei((X \setminus J) \cap I') \leq \beta \alpha(G',\wei) = \beta \wei(I')$.
In particular, $\alpha(G'-X,\wei) \geq (1-\beta)\alpha(G',\wei)$.
By a standard induction on the depth of the subtree of the recursion tree,
we obtain that if the depth of the recursion tree of a call $(G',s')$
is $h$, then $\wei(I_{G',s'}) \geq (1-\beta)^h \alpha(G',\wei)$. 
In particular, since the recursion depth at the root is bounded
by $s + \lfloor \log_{6/5} n \rfloor$, the root call returns
an independent set of weight at least 
\[ (1-\beta)^{s + \lfloor \log_{6/5} n \rfloor} \alpha(G,\wei) \geq \left(1-\beta(s + \lfloor \log_{6/5}n \rfloor)\right) \alpha(G,\wei) \geq (1-\varepsilon)\alpha(G,\wei). \]

Thus, it remains to prove Claim~\ref{clm:set_X}.

\subparagraph*{Proof of \cref{clm:set_X}.}
	Let $I'$ be an independent set in $G'$ of weight $\alpha(G',\wei)$.

	We start with a standard application of~\cref{lem:delete-heavy}.
	We enumerate the family $\mathcal{I}_0$ of
	$n^{\Oh(\log n / \gamma)}$ independent sets in $G'$ such that
	there exists $I_0 \in \mathcal{I}_0$ such that 
	$I_0 \subseteq I'$ and every $\gamma$-heavy
	vertex with respect to $I'$ belongs to $N[I_0]$.
	Henceforth, we will refer to this choice of $I_0$ as \emph{the correct choice of $I_0$}.

	Initiate $\mathcal{X} = \varnothing$.
	We iterate over the elements of $\mathcal{I}_0$. 
	For every $I_0 \in \mathcal{I}_0$ we proceed as follows.
	Let $X_0 \coloneqq N_{G'}[I_0]$ and $G_0 \coloneqq G'-X_0$.

	We find a shortest long hole in $G_0$ by~\cref{lem:find-hole}.
	If \cref{lem:find-hole} reports that $G_0$ is $C_{\geq t}$-free,
	we invoke the exact quasipolynomial-time algorithm of~\cite{DBLP:conf/stoc/GartlandLPPR21}
	that finds in time $n^{\Oh(\log^4 n)}$ an independent set $J$ of maximum
	weight in $G_0$. We insert $(V(G'), I_0 \cup J)$ into $\mathcal{X}$ and conclude
	computation for $I_0$. Note that for the correct choice of $I_0$,
	we have $\wei(I_0 \cup J) = \alpha(G',\wei)$.
	
	Assume then that \cref{lem:find-hole} returns a long hole $H$.
	If $|V(H)|\leq 2t+8$, then, assuming the correct choice of $I_0$,
	we have \[\wei(N_{G_0}[V(H)] \cap I') \leq (2t+8)\gamma \cdot \wei(I') \leq \beta\wei(I').\]
	Hence, we can insert $(X \coloneqq X_0 \cup N_{G_0}[V(H)], I_0)$ into $\mathcal{X}$
	and conclude the computation for the current $I_0$, as then $G'-X$ is non-adjacent to $H$,
	so every connected component of $G'-X$ satisfies Property~\ref{p:hole}.
	
	Therefore, from now on, we assume that $|V(H)|>2t+8$.
	
	\begin{claim}\label{clm:attach-to-hole}
		Every vertex $v \in V(G_0) \setminus V(H)$ has its neighbors in $V(H)$
		included in a~3-vertex subpath of $H$, or has a~neighbor in every subpath of $H$ at~least $t-1$ vertices.
	\end{claim}
	\begin{claimproof}
		Assume $v$ fails at realizing the latter condition.
		Then there is a~$(t-1)$-vertex subpath $P$ of~$H$, disjoint from the neighborhood of~$v$.
		Let $u \neq u' \in V(H)$ be such that $u$ and $u'$ are neighbors of $v$, and delimit a subpath $P'$ of $H$ containing $P$ and containing only two neighbors of $v$ (its endpoints $u$ and $u'$).
		Observe that if $v$ has at~most one neighbor in $V(H)$, it readily satisfies the first condition of the claim.
		For $V(P') \cup \{v\}$ not to contradict that $H$ is a~shortest  long hole of $G$, vertices $u$ and $u'$ have to be at distance at most~2 in $H$ (along the \emph{other} subpath of $H$ delimited by $u$ and $u'$), and so $v$ satisfies the first condition.
	\end{claimproof}
	
	We first get rid of the vertices satisfying the second condition of~\cref{clm:attach-to-hole}.
	Let $P$ be a~subpath of $H$ with $t-1$ vertices.
	We set $X_1 \coloneqq N_{G_0}(P) \setminus V(H)$ and 
	let $G_1$ be the connected component of $G_0 \setminus X_1$
	that contains $H$.

	We will insert $X_0 \cup X_1$ into any output set $X$ in what follows,
	so one may think of this step as deleting the vertices of $X_1$.
	Every connected component of $G_0 \setminus X_1$ except for $G_1$
	satisfies Property~\ref{p:hole}, so we need only to focus on $G_1$.

	As $|V(P)|=t-1$, assuming the correct choice of $I_0$,
	we have \[\wei(X_1 \cap I') \leqslant (t-1)\gamma \wei(I')\leq  \frac{\beta}{5}\wei(I').\]
    Note that every vertex satisfying the second condition of~\cref{clm:attach-to-hole}
	lies in $X_1$.

	Let us now consider BFS layering from $H$ in $G_1$ with layers
  $V(H)=L_0,\ L_1,\  L_2, \dots, L_h$.
	If $h \leq \lceil \frac{5}{\beta} \rceil$, we set $h' \coloneqq h + 1$,
	and otherwise we iterate over all options of $1 \leq h' \leq \lceil \frac{5}{\beta} \rceil$.
	In both cases, there exists a choice of $h'$ with 
	$\wei(I' \cap L_{h'}) \leq \frac{\beta}{5} \wei(I')$;
	we call it henceforth \emph{the correct choice of $h'$}.
	For fixed $h'$, we set $X_2 \coloneqq L_{h'}$ and
	$G_2 \coloneqq G_1[\bigcup_{i=0}^{h'-1} L_i]$.
	Note that $G_2$ is connected and
	every connected component of $G_1-X_2$ distinct from $G_2$
	lies in $\bigcup_{i > h'} L_i$ and, consequently, is non-adjacent to $H$
	and therefore satisfies point~\ref{p:hole}.
	Hence, in what follows we will insert $X_2$
	into any returned set $X$ and we focus on $G_2$.
	
	Note that $\cL \coloneqq L_0,L_1,\ldots,L_{h'-1}$ is the BFS
	layering in $G_2$ from $H$.
	The following claims and definitions refer to $G_2$.
	
For $i < j$, a~\emph{vertical path} between $u \in L_i$ and $v \in L_j$ with respect to $\cL$ is a~path $P$ such that $V(P)$ intersects exactly once every layer $L_k$ such that $k \in [i,j]$, and no other layer.
The \emph{cone} of a~vertex $v \in V(G_2)$ (still with respect to $\cL$) is the set of vertices \[ \{ x \in V(H) \mid \text{there is a~vertical path between $x$ and~$v$}\}.\] 
	
	\begin{claim}\label{clm:bounded_cones_length-pre}
		In $G_2$, let $x\in L_p$, $y\in L_q$, with $p\leq q$ such that there exists a path $P \subseteq \bigcup_{i\geq p} L_i$ with $\ell$ vertices, connecting $x$ and $y$ which inner vertices are non-adjacent to $H$.
     Let $x_1$ and $y_1$ be two vertices in the cones of $x$ and $y$, respectively.
		 Then, the distance between $x_1$ and $y_1$ along $H$ is at most $p+q+\ell$.
	\end{claim}			
	\begin{claimproof}
		Let $x$, $y$, $L_p$, $L_q$, $\ell$, $x_1$, $y_1$ and $P$ be as in the claim statement.
		
		Let $Q_x$ and $Q_y$ be two vertical paths from $x$ and $y$ 
    to $x_1$ and $y_1$, respectively.
    Let $x_2$ and $y_2$ be the penultimate vertices of $Q_x$ and $Q_y$,
    respectively. Note that $x_2,y_2 \in L_1$ and if $p=1$, then $x_2=x$ and $y_2=y$.
		Let $R$ be the shortest path in $G_2[V(P) \cup V(Q_x) \cup V(Q_y) \setminus \{x_1,y_1\}]$
		connecting $x_2$ and $y_2$. Note that $R$ is non-adjacent to $H$
		except for the endpoints.
		Furthermore, $|E(R)| \leq (p-1) + (q-1) + \ell$.

		Let $A$ and $B$ be the two subpaths of $H$ connecting $x_1$ and $y_1$.
		Without loss of generality, we assume that $|E(A)| \geq |E(B)|$.
		As $|V(H)| > 2t+8$, we have $|E(A)| \geq t + 5$.

		Let $x_3$ and $y_3$ be the neighbors of $x_2$ and $y_2$	on $A$, respectively, such that no vertex of $A$ between $x_3$ and $y_3$ 
		is adjacent to neither $x_2$ nor $y_2$.
		Let $A'$ be the subpath of $A$ from $x_3$ to $y_3$.
		Since the neighborhood in $H$ of 
    a vertex from $L_1$ is contained in a subpath of at most three consecutive vertices (by \cref{clm:attach-to-hole}),
		we have $|E(A')| \geq |E(A)|-4 \geq t$.

		We now consider a cycle $H'$ that is a concatenation of $x_2x_3$, $A'$, $y_3y_2$, and $R$.
		Note that $H'$ is an induced cycle in $G_2$. 
		As $|E(A')| \geq t$, $H'$ is a long hole.
		As $H$ is the shortest hole of length at least $t$, we have
		\[|E(A)| + |E(B)| = |E(H)| \leq |E(H')| = 2 + |E(A')| + |E(R)| \leq |E(A)| + p + q + \ell.\]
		Hence, $|E(B)| \leq p + q + \ell$. This completes the proof.
	\end{claimproof}
	\begin{claim}\label{clm:bounded_cones_length}
		In $G_2$, let $x\in L_p$, $y\in L_q$, with $p\leq q$ such that there exists a path $P \subseteq \bigcup_{i\geq p} L_i$ with $\ell$ vertices
    connecting $x$ and $y$ which is disjoint and non-adjacent to $H$. Then the union of cones of $x$ and $y$ is contained in a subpath of $H$ with at most $2(p+q+\ell)$ vertices.
	\end{claim}
	\begin{claimproof}
		Let $Z$ be the union of the cones of $x$ and $y$.
		By~\cref{clm:bounded_cones_length-pre}, any two vertices, one belonging to the cone of $x$ and another to the cone of $y$ respectively,
		in $Z$ lie within distance at most $p+q+\ell$ along $H$.
		Hence, $Z$ lies in the subpath of $H$ with at most $2(p+q+\ell)$ vertices.
	\end{claimproof}

	In the further proof, we use the following two corollaries of \cref{clm:bounded_cones_length}:
	
	\begin{claim}\label{cor:small_cone}
    For any $p$, the cone of any vertex $v \in L_p$ is contained in a~subpath of $H$ with at most $4p$ vertices.
	\end{claim}

	\begin{claim}\label{cor:small-cone-edge}
		Let $v_1 \in L_p$ and $v_2 \in (L_{p-1} \cup L_p) \cap N(v_1)$.
		Then the union of the cones of $v_1$ and $v_2$ is contained in a~subpath of $H$ with at most $4p+1$ vertices.
	\end{claim}

	For every $u \in V(H)$, let $D(u) \subseteq V(G_2)$ be the vertices whose cone contains~$u$.
	Let $u_0, u_1, \ldots, u_{|V(H)|-1}$ be a~numbering of the vertices of $H$ along the hole $H$.
	We set 
	\[
    D_{a,b} \coloneqq \bigcup_{k \in [a,a+b-1]} D(u_{k \bmod |V(H)|}) \quad \text{ and } b \coloneqq 4 \left\lceil \frac{5}{\beta} \right\rceil + 1.
	\]
	As $p\leq h'\leq \left\lceil \tfrac{5}{\beta} \right\rceil$, so $4p < b$. The following claim is a direct consequence of~\cref{cor:small_cone}.
	\begin{claim}\label{clm:disjoint-Dab}
		Two sets $D_{a,b}$ and $D_{a',b}$ are disjoint if $|a-a'| \bmod |V(H)| > 2b$.
	\end{claim}
	
	The removal of $D_{a,b} \cup D_{a',b}$ disconnects~$G$, provided $D_{a,b}$ and $D_{a',b}$ are disjoint.
	
	\begin{claim}\label{clm:two-strata-disconnecting}
		Let $a, a' \in [0,|V(H)|-1]$ be such that ${|a-a'| \bmod |V(H)|} > 2 b$.
		Then the removal of $D_{a,b} \cup D_{a',b}$ disconnects~$\tilde{G}$.
		In particular $u_{a+b \bmod |V(H)|}$ and $u_{a'+b \bmod |V(H)|}$ are in distinct connected components of $\tilde{G}-(D_{a,b} \cup D_{a',b})$.
	\end{claim}
	\begin{claimproof}
		Let $G_3 \coloneqq G_2-(D_{a,b} \cup D_{a',b})$, and assume without loss of generality that $a < a'$.
		We want to argue that $G_3$ is disconnected.
		By~\cref{cor:small-cone-edge}, no vertex of $V(G_3) \cap \bigcup_{j \in [a+b,a'-1]} D(u_j)$ can be adjacent to a~vertex of $V(G_3) \cap \bigcup_{j \in [0,a-1] \cup [a'+b,|V(H)|-1]} D(u_j)$.
		Thus in particular $u_{a+b}$ and $u_{a'+b \bmod |V(H)|}$ are in distinct connected components of~$G'$.
	\end{claimproof}
	
	Let $n_2 \coloneqq |V(G_2)|$.
	Fix $0 \leqslant a_1 \leqslant a_2 \leqslant |V(H)|-1$ minimizing $a_2-a_1$ subject to
  \begin{align}\label{eq:count}
    \left|\bigcup_{j \in [a_1,a_2]} D(u_j)\right| \geqslant \frac{n_2}{6}.
  \end{align}
  Observe that $a_1=a_2$ or 
  \begin{align}\label{eq:count2}
  \left|\bigcup_{j \in [a_1,a_2]} D(u_j)\right| < \frac{n_2}{3}.
  \end{align}

	We set $	d \coloneqq 2b \cdot \left\lceil \frac{5}{\beta} \right\rceil$, and distinguish two cases: $a_2-a_1 \leqslant d$ or $a_2-a_1 > d$.
	
	\medskip
	
\noindent	\textbf{Case: \boldmath $a_2-a_1 \leqslant d$.} Let us consider the family of sets \[
	\mathcal F_3 \coloneqq \{D_{a_1-d,b}, D_{a_1-d+2b,b}, D_{a_1-d+4b,b}, \ldots, D_{a_1-4b,b}, D_{a_1-2b,b}\}.
	\]
	By~\cref{clm:disjoint-Dab}, the $\lceil \frac{5}{\beta} \rceil$ sets of $\mathcal F_3$ are pairwise disjoint, thus there exists $X_3 \in \mathcal F_3$ such that $\wei(X_3 \cap I') \leqslant \frac{\beta}{5} \wei(I')$.
	For the same reason, there exists 
	\[
	X_4 \in \mathcal F_4 \coloneqq \{D_{a_2+1,b}, D_{a_2+1+2b,b}, D_{a_2+1+4b,b}, \ldots, D_{a_2+d+1-4b,b}, D_{a_2+d+1-2b,b}\}
	\] such that $\wei(X_4 \cap I') \leqslant \frac{\beta}{5} \wei(I')$.
	We iterate over all choices of $X_3$ and $X_4$.
	Finally, we exploit the fact that $a_2-a_1$ is small by adding $X_5 \coloneqq \bigcup_{j \in [a_1,a_2]} N(u_j)$ to the set of vertices to remove.
	As $d \cdot \gamma \leqslant \frac{\beta}{5}$, for the correct choices
	of $I_0$ and $h'$ we have $\wei(X_5 \cap I') \leqslant \frac{\beta}{5} \wei(I')$.
	We define $X \coloneqq X_0 \cup X_1 \cup X_2 \cup X_3 \cup X_4 \cup X_5$
	and insert $(X,I_0)$ into $\mathcal{X}$.
	We have $\wei(X \cap (I' \setminus I_0)) \leq \beta \wei(I')$.

	We already argued that the connected components of $G'-X$ disjoint from $V(G_2)$ satisfy Property~\ref{p:hole}.
	By~\cref{clm:two-strata-disconnecting}, $G_2-(X_3 \cup X_4 \cup X_5)$ has two sets of connected components: those intersecting $\bigcup_{j \in [a_1,a_2]} D(u_j)$, and those not.
	The former kind are not adjacent to $H$, so satisfy Property~\ref{p:hole}.
  The latter kind has at~most $5n_2/6 \leqslant 5|V(G')|/6$ vertices, by design of the interval $[a_1,a_2]$ using \eqref{eq:count},
	so satisfy Property~\ref{p:size}.
	
	\medskip
\noindent\textbf{Case: \boldmath $a_2-a_1 > d$.} Recall that by \eqref{eq:count2} we have $|\bigcup_{j \in [a_1,a_2]} D(u_j)| < n_2/3$.
	Then $\bigcup_{j \in [a_1-d,a_1-1]} D(u_j)$ and $\bigcup_{j \in [a_2+1,a_2+d]} D(u_j)$ contains each less than $n_2/6$ vertices, by the minimality of~$a_2-a_1$.
	Hence $|\bigcup_{j \in [a_1-d,a_2+d]} D(u_j)|<2n_2/3$.
	By the pigeonhole principle, as $d = 2b \cdot \lceil \frac{5}{\beta} \rceil$,
	there is $a \in \{a_1-d,a_1-d+2b, \ldots,a_1-2b\}$ (resp.,~$a' \in \{a_2+2b,a_2+4b,\ldots,a_2+d\}$) 
	such that $\wei(D_{a,b} \cap I') \leqslant \frac{\beta}{5} \wei(I')$ 
	(resp.,~$\wei(D_{a',b} \cap I') \leqslant \frac{\beta}{5} \wei(I')$).
	We iterate over all choices of $a$ and $a'$
	and set $X_3 \coloneqq D_{a,b}$, $X_4 \coloneqq D_{a',b}$.
	We take $X \coloneqq X_0 \cup X_1 \cup X_2 \cup X_3 \cup X_4$
	and insert $(X,I_0)$ into $\mathcal{X}$.
	For the correct choice of $I_0$, $h'$, $a$, $a'$
	we have $\wei(X \cap (I' \setminus I_0)) \leq \beta \wei(I')$.
	
	Observe finally that $G_2-(X_3 \cup X_4)$ has
	no connected component of size larger than $2n_2/3$, hence larger than $5|V(G')|/6$.
	This completes the proof of~\cref{clm:set_X} and thus, of~\cref{thm:mwis}.
\end{proof}

\section{\boldmath QPTAS for $(\tw \leq r,\psi)$-\MWIS: Proof of \cref{thm:pack}}\label{sec:pack}
The proof of \cref{thm:pack} is based on the approach introduced by Gartland et al.~\cite[Section~4]{DBLP:conf/stoc/GartlandLPPR21}.
For a graph $G$, we define its \emph{blob graph} $G^\circ$ as follows:
\begin{align*}
V(G^\circ) = & \{ X \subseteq V(G) ~|~G[X] \text{ is connected}\}\\
E(G^\circ) = & \{ XY ~|~ G[X \cup Y] \text{ is connected}\}.
\end{align*}
Equivalently, edges join sets that are either non-disjoint, or there is an edge from one set to another.
Gartland et al.~\cite{DBLP:conf/stoc/GartlandLPPR21} proved that for every $t \geq 4$, a graph $G$ is $C_{\geq t}$-free if and only if $G^\circ$ is $C_{\geq t}$-free.
First, let us extend this result to the setting of $H$-induced-minor-free graphs, where every component of $H$ is a hole.
The proof closely follows the proof of Paesani et al. for the case if $H$ is a linear forest~\cite{DBLP:journals/siamdm/PaesaniPR22}.

\begin{theorem}\label{thm:blobimfree}
Let $H$ be a graph whose every component is a hole.
The graph $G$ contains $H$ as an induced minor if and only if $G^\circ$ contains $H$ as an induced minor.
\end{theorem}
\begin{proof}
As $G$ is an induced subgraph of $G^\circ$, the forward implication is immediate.
We prove the backward implication by induction on the number $s$ of connected components of $H$.
The case $s=1$ was shown by Gartland et al.~\cite{DBLP:conf/stoc/GartlandLPPR21}.
Thus, assume that $s > 1$ and let $C$ be a~connected component of $H$.
Let $H' = H - C$.

Suppose that $G^\circ$ contains $H$ as an induced minor.
Fix one induced minor model of $H$ in $G$ and let $\cX$ be the set of vertices of the model.
This means that $G^\circ[\cX]$ has $s$ components, each of which is a cycle,
and there is a one-to-one mapping from the components of $G^\circ[\cX]$ and components of $H$,
so that each cycle is mapped to a cycle of at most its own length.

Let $\cY \subseteq \cX$ be the vertices of $G^\circ$ that form the cycle mapped to $C$;
each element of $\cY$ is a subset of $V(G)$.
Let $Y \subseteq V(G)$ be the union of all sets in $\cY$.
Note that $G^\circ[\cY]$ is an induced subgraph of $(G[Y])^\circ$.
Thus, by the inductive assumption, $G[Y]$ contains an induced cycle with at least $|V(C)|$ vertices.

Now let $X \subseteq V(G)$ be the union of sets in $\cX \setminus \cY$.
Since $\cY$ is the vertex set of one component of $G^\circ[\cX]$, there are no edges between $\cX \setminus \cY$ and $\cY$.
Consequently, there are no edges in $G$ between $X$ and $Y$, and thus $G^\circ[X \setminus Y]$ is an induced subgraph of $(G-N[Y])^\circ$.

Since $G^\circ[\cX \setminus \cY]$, and thus $(G - N[Y])^\circ$, contains $H'$ as an induced minor,
by the inductive assumption we know that $G-N[Y]$ contains $H'$ as an induced minor.
Combining its model with the model of $C$ in $G[Y]$, we obtain an induced minor model of $H$ in $G$.
\end{proof}

Let $\mathcal{H}$ be a class of graphs.
An \emph{induced $\mathcal{H}$-packing} in a graph $G$ is a set $S \subseteq V(G)$ such that every component of $G[S]$ belongs to $\mathcal{H}$.
For example, if $\cH = \{ K_1 \}$, then induced $\mathcal{H}$-packing in $G$ if and only if it is an independent set.

The construction of blob graphs allows us to reduce the problem of finding an induced $\cH$-packing of maximum size (or weight) in $G$ to solving the \MWIS problem in an appropriate induced subgraph of $G^\circ$.
This is expressed in the following lemma whose proof is immediate, see e.g.,~\cite[Section 4]{DBLP:conf/stoc/GartlandLPPR21}.

\begin{lemma}\label{lem:useblobs}
Let $G$ be a graph, $\wei : V(G) \to \RR_{\geq 0}$ be a weight function, and $\cH$ be any class of graphs.
Let $\cX = \{ X \subseteq V(G) ~|~ G[X] \in \cH \text{ and is connected}\}$.
Let $\wei^\circ : \cX \to \RR_{\geq 0}$ be defined as $\wei^\circ(X) = \sum_{v \in X} \wei(v)$.
Then the maximum weight of an induced $\cH$-packing in $G$ is equal to $\alpha(G^\circ[\cX],\wei^\circ)$.
\end{lemma}

The problem with using \cref{lem:useblobs} is that $\cX$ might be very large, so the running time of the algorithm solving \MWIS on $(G^\circ[\cX],\wei^\circ)$ is not bounded by a moderate function of $|V(G)|$.

A class $\cH$ of graphs is \emph{weakly hyperfinite} if for every $\delta > 0$ there exists $c_\delta \in \NN$ such that for any $G \in \cH$ there exists a set $S \subseteq |V(G)|$ of size at most $\delta \cdot |V(G)|$ so that every component of $G - S$ has at most $c_\delta$ vertices~\cite[Section 16.2]{DBLP:books/daglib/0030491}.
The following theorem follows the result of Gartland et al.~\cite{DBLP:conf/stoc/GartlandLPPR21}. Again, we include the proof for completeness.

\begin{theorem}\label{thm:hyperfinite}
Let $\mathcal{H}$ be a non-empty hereditary weakly hyperfinite class of graphs.
Let $s,t$ be positive integers and $\varepsilon \in (0,1)$ be a real.
There is an algorithm that, given a graph $G$, in quasipolynomial time returns one of the following outputs:
\begin{itemize}
\item an induced minor model of $sC_t$ in $G$, or
\item an induced $\mathcal{H}$-packing of size at least $(1-\varepsilon)$ times the size of a largest induced $\mathcal{H}$-packing in $G$.
\end{itemize}
\end{theorem}
\begin{proof}
Without loss of generality assume that $t \geq 4$.
Let $\delta = \varepsilon/2$ and let $c_{\delta}$ be the constant witnessing that $\cH$ is weakly hyperfinite.
Let $G$ be an $n$-vertex instance instance of the problem.
Let 
\[
\cX = \{ X \subseteq V(G) ~|~ |X| \leq c_\delta \text{ and } G[X] \text{ is a connected graph from } \cH\}.
\]
Consider the graph $G^\circ[\cX]$; note that it has $\Oh(n^{c_\delta})$ vertices and can be constructed in time polynomial in $n$ as $c_\delta$ is a constant.
We call the algorithm from \cref{thm:mwis} for $G^\circ[\cX]$, weigths $\wei^\circ$ defined as in \cref{lem:useblobs},
and precision $\varepsilon/2$. It running time is quasipolynomial in  $\Oh(n^{c_\delta})$ and thus in $n$.

If an induced minor of $sC_t$ is reported, we report an induced minor of $sC_t$ in $G$.
Indeed, recall that $G^\circ[\cX]$ is an induced subgraph of $G^\circ$.
Consequently, by \cref{thm:blobimfree}, if $G^\circ[\cX]$ contains $sC_t$ as an induced minor, so does $G$.

Thus, suppose that the algorithm from \cref{thm:mwis} returns an induced $\cH$-packing $\cS$ in $G^\circ[\cX]$.
It is straightforward to verify that $S = \bigcup \cS$ is an induced $\cH$-packing in $G$. Let us argue how its size compares to the optimum one.

Let $S^*$ be an optimum induced $\cH$-packing in $G$.
Let us construct another induced $\cH$-packing $S'$ as follows.
We consider every component of $G[S^*]$ separately, let $C$ be the vertex set of one such component.
If $|C| \leq c_\delta$, we set $C' = C$ and include this set in $S'$.
Otherwise, since $G[C]$ is weakly hyperfinite, there is a set $Q_C \subseteq C$ of size at most $\delta |C| = \varepsilon/2 \cdot |C|$, such that every component of $G[C \setminus Q_C]$ has at most $c_\delta$ vertices.
We set $C' = C \setminus Q_C$ and include it into $S'$.
Since $\cH$ is hereditary, it holds that  $G[C'] \in \cH$.
Note that in both cases we have $|C'| \geq  (1-\delta)|C|= (1-\varepsilon/2)|C|$. Consequently, we obtain
\[
	|S^*| = \sum_{C: \text{ component of } G[S^*]} |C| \leq \sum_{C: \text{ component of } G[S^*]} \frac{|C'|}{(1-\varepsilon/2)} = \frac{|S'|}{(1-\varepsilon/2)}.
\]
Now, by \cref{lem:useblobs} notice that the size of $S'$ is equal to the weight of a largest-weight independent set in $G^\circ[\cX]$.
Thus, by \cref{thm:mwis}, we have $|S| \geq (1-\varepsilon/2)|S'|$.
Summarizing, it holds that
\[
	|S| \geq (1-\varepsilon/2)|S'| \geq (1-\varepsilon/2)(1-\varepsilon/2) |S^*| = (1-\varepsilon + \varepsilon^2)|S^*| \geq (1-\varepsilon)|S^*|.
\]
This completes the proof.
\end{proof}

Now let us argue how \cref{thm:hyperfinite} implies \cref{thm:pack}.

\thmpack*
\begin{proof}
Note that if $\psi$ is \emph{consistent}, i.e., if there is any graph $H$ such that $H \models \psi$,
then, since $\psi$ is hereditary, we have $K_1 \subseteq \psi$.
Consequently, in such a case, every non-empty instance of $(\tw \leq r,\psi)$-\MWIS has a solution (for example, a single vertex).
Since $\psi$ is a fixed formula, we can verify if it is consistent in constant time.
If not, we return the third outcome.

So from now on let us assume that some solution exists.
It is known that graphs of treewidth at most $r$, for any fixed $r$, form a weakly hyperfinite class of graphs~\cite[Theorem 16.5]{DBLP:books/daglib/0030491}.
Furthermore, this class is closed under vertex deletion and disjoint unions.
Consequently, graphs of treewidth at most $r$ satisfying $\psi$ form a non-empty class closed under vertex deletion and disjoint unions.
Thus, we obtain \cref{thm:pack} as an immediate corollary of \cref{thm:hyperfinite}.
\end{proof}

\section{Conclusion}

Let us conclude the paper with pointing out two specific problems for further research.
First, it would be interesting to strengthen our \cref{thm:mwis} by solving the problem \emph{exactly} without significantly increasing the running time.

\begin{problem}
Show that for every fixed $s,t$, \MWIS can be solved in quasipolynomial time in $sC_t$-induced-minor-free graphs.
\end{problem}

Second, let us suggest the following possible extension of our \cref{thm:mwis}.
Our starting point was a quasipolynomial-time algorithm for $H$-induced-minor-free graphs and we aimed to extend this result for $H'$-induced-minor-free graphs, where every component of $H$ is isomorphic to $H$.
Perhaps this can be done in a more general setting?

\begin{problem}\label{prob:components}
Suppose that $H_1$ and $H_2$ are graphs such that \MWIS is (quasi)polynomial-time solvable in $H_i$-induced-minor-free graphs for $i \in \{1,2\}$.
Show a quasipolynomial-time algorithm (or at least a QPTAS) for \MWIS in $(H_1 + H_2)$-induced-minor-free graphs.
\end{problem}

We remark that the solution to Problem~\ref{prob:components} is known if instead of forbidding certain graphs as induced minors, we forbid them as induced subgraphs~\cite{DBLP:conf/focs/GartlandL20}.
However, the approach relies on the fact that induced subgraphs are constant-size objects, while an induced minor model of a constant-size graph might be arbitrarily large.

\bibliography{bibliography}

@article{grzesik2020polynomialtime,
  author       = {Andrzej Grzesik and
                  Tereza Klimo\v{s}ov{\'{a}} and
                  Marcin Pilipczuk and
                  Michał Pilipczuk},
  title        = {Polynomial-time Algorithm for Maximum Weight Independent Set on {$P_6$}-free
                  Graphs},
  journal      = {{ACM} Trans. Algorithms},
  volume       = {18},
  number       = {1},
  pages        = {4:1--4:57},
  year         = {2022},
  url          = {https://doi.org/10.1145/3414473},
  doi          = {10.1145/3414473},
  bibsource    = {dblp computer science bibliography, https://dblp.org}
}

@article{Tara,
  author       = {Tara Abrishami and
                  Maria Chudnovsky and
                  Marcin Pilipczuk and
                  Paweł Rzążewski and
                  Paul D. Seymour},
  title        = {Induced Subgraphs of Bounded Treewidth and the Container Method},
  journal      = {{SIAM} J. Comput.},
  volume       = {53},
  number       = {3},
  pages        = {624--647},
  year         = {2024},
  url          = {https://doi.org/10.1137/20m1383732},
  doi          = {10.1137/20M1383732},
  bibsource    = {dblp computer science bibliography, https://dblp.org}
}

@inproceedings{DBLP:conf/stoc/GartlandLMPPR24,
  author       = {Peter Gartland and
                  Daniel Lokshtanov and
                  Tom{\'{a}}s Masa\v{r}{\'{\i}}k and
                  Marcin Pilipczuk and
                  Michał Pilipczuk and
                  Paweł Rzążewski},
  editor       = {Bojan Mohar and
                  Igor Shinkar and
                  Ryan O'Donnell},
  title        = {Maximum Weight Independent Set in Graphs with no Long Claws in Quasi-Polynomial
                  Time},
  booktitle    = {Proceedings of the 56th Annual {ACM} Symposium on Theory of Computing,
                  {STOC} 2024, Vancouver, BC, Canada, June 24-28, 2024},
  pages        = {683--691},
  publisher    = {{ACM}},
  year         = {2024},
  url          = {https://doi.org/10.1145/3618260.3649791},
  doi          = {10.1145/3618260.3649791},
  bibsource    = {dblp computer science bibliography, https://dblp.org}
}

@inproceedings{Karp1972,
author="Karp, Richard M.",
editor="Miller, Raymond E.
and Thatcher, James W.
and Bohlinger, Jean D.",
title="Reducibility among Combinatorial Problems",
bookTitle="Complexity of Computer Computations: Proceedings of a symposium on the Complexity of Computer Computations, held March 20--22, 1972, at the IBM Thomas J. Watson Research Center, Yorktown Heights, New York, and sponsored by the Office of Naval Research, Mathematics Program, IBM World Trade Corporation, and the IBM Research Mathematical Sciences Department",
year="1972",
publisher="Springer US",
address="Boston, MA",
pages="85--103",
isbn="978-1-4684-2001-2",
doi="10.1007/978-1-4684-2001-2_9",
url="https://doi.org/10.1007/978-1-4684-2001-2_9"
}

@inproceedings{DBLP:conf/focs/GartlandL20,
  author       = {Peter Gartland and
                  Daniel Lokshtanov},
  editor       = {Sandy Irani},
  title        = {Independent Set on {$P_k$}-Free
                  Graphs in Quasi-Polynomial Time},
  booktitle    = {61st {IEEE} Annual Symposium on Foundations of Computer Science, {FOCS}
                  2020, Durham, NC, USA, November 16-19, 2020},
  pages        = {613--624},
  publisher    = {{IEEE}},
  year         = {2020},
  url          = {https://doi.org/10.1109/FOCS46700.2020.00063},
  doi          = {10.1109/FOCS46700.2020.00063},
  bibsource    = {dblp computer science bibliography, https://dblp.org}
}

@inproceedings{DBLP:conf/sosa/PilipczukPR21,
  author       = {Marcin Pilipczuk and
                  Michał Pilipczuk and
                  Paweł Rzążewski},
  editor       = {Hung Viet Le and
                  Valerie King},
  title        = {Quasi-polynomial-time algorithm for Independent Set in {$P_t$}-free
                  graphs via shrinking the space of induced paths},
  booktitle    = {4th Symposium on Simplicity in Algorithms, {SOSA} 2021, Virtual Conference,
                  January 11-12, 2021},
  pages        = {204--209},
  publisher    = {{SIAM}},
  year         = {2021},
  url          = {https://doi.org/10.1137/1.9781611976496.23},
  doi          = {10.1137/1.9781611976496.23},
  bibsource    = {dblp computer science bibliography, https://dblp.org}
}

@inproceedings{DBLP:conf/stoc/GartlandLPPR21,
  author       = {Peter Gartland and
                  Daniel Lokshtanov and
                  Marcin Pilipczuk and
                  Michał Pilipczuk and
                  Paweł Rzążewski},
  editor       = {Samir Khuller and
                  Virginia Vassilevska Williams},
  title        = {Finding large induced sparse subgraphs in {$C_{>t}$}-free graphs in quasipolynomial time},
  booktitle    = {{STOC} '21: 53rd Annual {ACM} {SIGACT} Symposium on Theory of Computing,
                  Virtual Event, Italy, June 21-25, 2021},
  pages        = {330--341},
  publisher    = {{ACM}},
  year         = {2021},
  url          = {https://doi.org/10.1145/3406325.3451034},
  doi          = {10.1145/3406325.3451034},
  bibsource    = {dblp computer science bibliography, https://dblp.org}
}

@inproceedings{DBLP:conf/soda/LokshantovVV14,
  author       = {Daniel Lokshtanov and
                  Martin Vatshelle and
                  Yngve Villanger},
  editor       = {Chandra Chekuri},
  title        = {Independent Set in {$P_5$}-Free Graphs in Polynomial
                  Time},
  booktitle    = {Proceedings of the Twenty-Fifth Annual {ACM-SIAM} Symposium on Discrete
                  Algorithms, {SODA} 2014, Portland, Oregon, USA, January 5-7, 2014},
  pages        = {570--581},
  publisher    = {{SIAM}},
  year         = {2014},
  url          = {https://doi.org/10.1137/1.9781611973402.43},
  doi          = {10.1137/1.9781611973402.43},
  bibsource    = {dblp computer science bibliography, https://dblp.org}
}

@inproceedings{DBLP:conf/soda/ChudnovskyMPPR24,
  author       = {Maria Chudnovsky and
                  Rose McCarty and
                  Marcin Pilipczuk and
                  Michał Pilipczuk and
                  Paweł Rzążewski},
  editor       = {David P. Woodruff},
  title        = {Sparse induced subgraphs in {$P_6$}-free graphs},
  booktitle    = {Proceedings of the 2024 {ACM-SIAM} Symposium on Discrete Algorithms,
                  {SODA} 2024, Alexandria, VA, USA, January 7-10, 2024},
  pages        = {5291--5299},
  publisher    = {{SIAM}},
  year         = {2024},
  url          = {https://doi.org/10.1137/1.9781611977912.190},
  doi          = {10.1137/1.9781611977912.190},
  bibsource    = {dblp computer science bibliography, https://dblp.org}
}

@article{DBLP:journals/siamdm/PaesaniPR22,
  author       = {Giacomo Paesani and
                  Dani{\"{e}}l Paulusma and
                  Pawe\l{} Rz\k{a}\.zewski},
  title        = {{Feedback Vertex Set} and {Even Cycle Transversal} for {$H$}-Free
                  Graphs: Finding Large Block Graphs},
  journal      = {{SIAM} J. Discret. Math.},
  volume       = {36},
  number       = {4},
  pages        = {2453--2472},
  year         = {2022},
  url          = {https://doi.org/10.1137/22m1468864},
  doi          = {10.1137/22m1468864},
  bibsource    = {dblp computer science bibliography, https://dblp.org}
}

@book{platypus,
  author       = {Marek Cygan and
                  Fedor V. Fomin and
                  Lukasz Kowalik and
                  Daniel Lokshtanov and
                  D{\'{a}}niel Marx and
                  Marcin Pilipczuk and
                  Michał Pilipczuk and
                  Saket Saurabh},
  title        = {Parameterized Algorithms},
  publisher    = {Springer},
  year         = {2015},
  url          = {https://doi.org/10.1007/978-3-319-21275-3},
  doi          = {10.1007/978-3-319-21275-3},
  isbn         = {978-3-319-21274-6},
  bibsource    = {dblp computer science bibliography, https://dblp.org}
}

@article{DBLP:journals/siamcomp/ChalermsookCKLM20,
  author       = {Parinya Chalermsook and
                  Marek Cygan and
                  Guy Kortsarz and
                  Bundit Laekhanukit and
                  Pasin Manurangsi and
                  Danupon Nanongkai and
                  Luca Trevisan},
  title        = {From {Gap-Exponential Time Hypothesis} to Fixed Parameter Tractable
                  Inapproximability: Clique, Dominating Set, and More},
  journal      = {{SIAM} J. Comput.},
  volume       = {49},
  number       = {4},
  pages        = {772--810},
  year         = {2020},
  url          = {https://doi.org/10.1137/18M1166869},
  doi          = {10.1137/18M1166869},
  bibsource    = {dblp computer science bibliography, https://dblp.org}
}

@inproceedings{DBLP:conf/focs/LinRSW23,
  author       = {Bingkai Lin and
                  Xuandi Ren and
                  Yican Sun and
                  Xiuhan Wang},
  title        = {Improved Hardness of Approximating $k$-{Clique} under {ETH}},
  booktitle    = {64th {IEEE} Annual Symposium on Foundations of Computer Science, {FOCS}
                  2023, Santa Cruz, CA, USA, November 6-9, 2023},
  pages        = {285--306},
  publisher    = {{IEEE}},
  year         = {2023},
  url          = {https://doi.org/10.1109/FOCS57990.2023.00025},
  doi          = {10.1109/FOCS57990.2023.00025},
  bibsource    = {dblp computer science bibliography, https://dblp.org}
}

@INPROCEEDINGS{Hastad96cliqueis,
    author = {Johan H{\aa}stad},
    title = {Clique is hard to approximate within $n^{{(1-\epsilon)}}$},
    booktitle = {Acta Mathematica},
    year = {1996},
    pages = {627--636},
    publisher = {}
}

@article{GAREY1976237,
title = "Some simplified {NP}-complete graph problems",
journal = "Theoretical Computer Science",
volume = "1",
number = "3",
pages = "237--267",
year = "1976",
issn = "0304-3975",
doi = "https://doi.org/10.1016/0304-3975(76)90059-1",
url = "http://www.sciencedirect.com/science/article/pii/0304397576900591",
author = "M.R. Garey and D.S. Johnson and L. Stockmeyer",
}

@article{alekseev1982effect,
  title={The effect of local constraints on the complexity of determination of the graph independence number},
  author={Alekseev, Vladimir E.},
  journal={Combinatorial-algebraic methods in applied mathematics},
  pages={3--13},
  year={1982}
}

@article{SBIHI198053,
title = "Algorithme de recherche d'un stable de cardinalite maximum dans un graphe sans etoile",
journal = "Discrete Mathematics",
volume = "29",
number = "1",
pages = "53--76",
year = "1980",
issn = "0012-365X",
doi = "https://doi.org/10.1016/0012-365X(90)90287-R",
_url = "http://www.sciencedirect.com/science/article/pii/0012365X9090287R",
author = "Najiba Sbihi",
}

@article{MINTY1980284,
title = "On maximal independent sets of vertices in claw-free graphs",
journal = "Journal of Combinatorial Theory, Series B",
volume = "28",
number = "3",
pages = "284--304",
year = "1980",
issn = "0095-8956",
doi = "https://doi.org/10.1016/0095-8956(80)90074-X",
_url = "http://www.sciencedirect.com/science/article/pii/009589568090074X",
author = "George J. Minty",
}

@article{DBLP:journals/jda/LozinM08,
  author    = {Vadim V. Lozin and
               Martin Milani\v{c}},
  title     = {A polynomial algorithm to find an independent set of maximum weight
               in a fork-free graph},
  journal   = {J. Discrete Algorithms},
  volume    = {6},
  number    = {4},
  pages     = {595--604},
  year      = {2008},
  url       = {https://doi.org/10.1016/j.jda.2008.04.001},
  doi       = {10.1016/j.jda.2008.04.001},
  bibsource = {dblp computer science bibliography, https://dblp.org}
}

@article{ALEKSEEV20043,
title = "Polynomial algorithm for finding the largest independent sets in graphs without forks",
journal = "Discrete Applied Mathematics",
volume = "135",
number = "1",
pages = "3--16",
year = "2004",
note = "Russian Translations II",
issn = "0166-218X",
doi = "10.1016/S0166-218X(02)00290-1",
_url = "http://www.sciencedirect.com/science/article/pii/S0166218X02002901",
  author={Alekseev, Vladimir E.},
}

@article{DBLP:journals/dm/Munaro17,
  author       = {Andrea Munaro},
  title        = {On line graphs of subcubic triangle-free graphs},
  journal      = {Discret. Math.},
  volume       = {340},
  number       = {6},
  pages        = {1210--1226},
  year         = {2017},
  url          = {https://doi.org/10.1016/j.disc.2017.01.006},
  doi          = {10.1016/J.DISC.2017.01.006},
  bibsource    = {dblp computer science bibliography, https://dblp.org}
}

@article{Po74,
author={Svatopluk Poljak},
title={A note on stable sets and colorings of graphs},
journal={Commentationes Mathematicae Universitatis Carolinae},
volume={15},
issue={2},
pages={307--309},
year={1974}
}

@book{DBLP:books/daglib/0030491,
  author       = {Jaroslav Ne\v{s}et\v{r}il and
                  Patrice Ossona de Mendez},
  title        = {Sparsity - Graphs, Structures, and Algorithms},
  series       = {Algorithms and combinatorics},
  volume       = {28},
  publisher    = {Springer},
  year         = {2012},
  url          = {https://doi.org/10.1007/978-3-642-27875-4},
  doi          = {10.1007/978-3-642-27875-4},
  isbn         = {978-3-642-27874-7},
  bibsource    = {dblp computer science bibliography, https://dblp.org}
}

@phdthesis{GartlandThesis,
  title        = {Quasi-Polynomial Time Techniques for Independent Set and Beyond in Hereditary Graph Classes},
  author       = {Peter Gartland},
  year         = 2023,  
  school       = {University of California Santa Barbara},
  type         = {{PhD thesis}}
}

@article{DBLP:journals/siamcomp/ChudnovskyPPT24,
  author       = {Maria Chudnovsky and
                  Marcin Pilipczuk and
                  Michal Pilipczuk and
                  St{\'{e}}phan Thomass{\'{e}}},
  title        = {Quasi-Polynomial Time Approximation Schemes for the Maximum Weight
                  Independent Set Problem in {$H$}-Free
                  Graphs},
  journal      = {{SIAM} J. Comput.},
  volume       = {53},
  number       = {1},
  pages        = {47--86},
  year         = {2024},
  url          = {https://doi.org/10.1137/20m1333778},
  doi          = {10.1137/20M1333778},
  bibsource    = {dblp computer science bibliography, https://dblp.org}
}

@article{DBLP:journals/algorithmica/BonnetDGTW26,
  author       = {{\'{E}}douard Bonnet and
                  Julien Duron and
                  Colin Geniet and
                  St{\'{e}}phan Thomass{\'{e}} and
                  Alexandra Wesolek},
  title        = {Maximum Independent Set when Excluding an Induced Minor: {$K_1 + tK_2$} and {$tC_3 \uplus C_4$}},
  journal      = {Algorithmica},
  volume       = {88},
  number       = {1},
  pages        = {16},
  year         = {2026},
  url          = {https://doi.org/10.1007/s00453-025-01356-2},
  doi          = {10.1007/S00453-025-01356-2},
  bibsource    = {dblp computer science bibliography, https://dblp.org}
}

@inproceedings{DBLP:conf/soda/BonamyBDEGHTW23,
  author       = {Marthe Bonamy and
                  {\'{E}}douard Bonnet and
                  Hugues D{\'{e}}pr{\'{e}}s and
                  Louis Esperet and
                  Colin Geniet and
                  Claire Hilaire and
                  St{\'{e}}phan Thomass{\'{e}} and
                  Alexandra Wesolek},
  editor       = {Nikhil Bansal and
                  Viswanath Nagarajan},
  title        = {Sparse graphs with bounded induced cycle packing number have logarithmic
                  treewidth},
  booktitle    = {Proceedings of the 2023 {ACM-SIAM} Symposium on Discrete Algorithms,
                  {SODA} 2023, Florence, Italy, January 22-25, 2023},
  pages        = {3006--3028},
  publisher    = {{SIAM}},
  year         = {2023},
  url          = {https://doi.org/10.1137/1.9781611977554.ch116},
  doi          = {10.1137/1.9781611977554.CH116},
  bibsource    = {dblp computer science bibliography, https://dblp.org}
}

@article{DBLP:journals/dam/BrandstadtM18a,
  author       = {Andreas Brandst{\"{a}}dt and
                  Raffaele Mosca},
  title        = {Maximum weight independent set for {$\ell$}claw-free graphs
                  in polynomial time},
  journal      = {Discret. Appl. Math.},
  volume       = {237},
  pages        = {57--64},
  year         = {2018},
  url          = {https://doi.org/10.1016/j.dam.2017.11.029},
  doi          = {10.1016/J.DAM.2017.11.029},
  bibsource    = {dblp computer science bibliography, https://dblp.org}
}

@article{dallard2022secondpaper,
  author       =  {Cl{\'{e}}ment Dallard and
                  Martin Milani{\v{c}} and
                  Kenny {\v{S}}torgel},
  title = {Treewidth versus clique number. {III}. {T}ree-independence number of graphs with a forbidden structure},
  journal = {Journal of Combinatorial Theory, Series B},
  volume = {167},
  pages = {338--391},
  year = {2024},
  doi = {10.1016/j.jctb.2024.03.005},
}

\end{document}